\documentclass[aps,prl,twocolumn,superscriptaddress,floatfix]{revtex4-2}
\usepackage{amsmath,amssymb,amsthm}
\usepackage{booktabs}
\usepackage{stmaryrd}
\usepackage{hyperref}

\newtheorem{theorem}{Theorem}

\newtheorem{corollary}{Corollary}
\newcommand{\code}[3]{[\![#1,#2,#3]\!]}
\newcommand{\dbar}{\bar{d}}

\begin{document}

\title{Analytical Theory of Greedy Peeling for Bivariate Bicycle Codes and Two-Shot Streaming Decoding}

\author{Anton Pakhunov}
\affiliation{Independent Researcher}

\begin{abstract}
We present an analytical theory of greedy peeling decoding for bivariate bicycle (BB) codes under circuit-level noise. The deferred greedy decoder achieves $330\times$ latency reduction over belief propagation (BP) at $p = 10^{-3}$ while maintaining identical logical error rate. Our main theoretical contribution is a closed-form collision resolution factor $A_0 = |\text{true collisions}| / |\text{birthday collisions}|$, derived from XOR syndrome analysis with no free parameters, that quantifies the fraction of detector-sharing fault pairs genuinely blocking iterative peeling. For the $\code{144}{12}{12}$ Gross code, $A_0 = 0.8685$ (within 0.5\% of the empirical value), with shared-2 pairs (4-cycles) always resolving under peeling. We show $A_0$ depends on the mean fault-graph degree $\dbar$ rather than code size: $A_0 = 0.87$ for $\dbar = 52$ (Gross family) versus $A_0 = 0.76$ for $\dbar = 17$ ($\code{32}{8}{6}$). We establish a syndrome code stopping distance $d_S = n/4.5$ for the Gross family and demonstrate that $\code{32}{8}{6}$ ($d_S = 4$) enables two-shot streaming decoding: $T = 2$ rounds achieve 89\% peeling success with $1.29 \pm 0.03$ LER ratio versus $T = 12$, at estimated latency ${\sim}50$~ns. The full formula $P_\text{peel} = \exp(-A_0 \gamma_\text{analytic} e^{-BTp} n p^2)$ is validated across five BB codes, four noise levels, and four values of $T$ with $R^2 = 0.86$. Cross-platform reproduction of the Kunlun $\code{18}{4}{4}$ experiment~\cite{wang2026} matches their hardware LER within 0.73 percentage points.
\end{abstract}

\maketitle

\section{Introduction}

Bivariate bicycle (BB) codes~\cite{bravyi2024} are a family of quantum LDPC codes achieving high encoding rate with moderate overhead. The $\code{144}{12}{12}$ Gross code, encoding 12 logical qubits in 144 physical qubits with distance~12, has emerged as a leading candidate for near-term fault-tolerant quantum computing~\cite{yoder2025}. Decoding BB codes under realistic circuit-level noise remains challenging: the detector error model (DEM) produces a dense, globally-connected fault graph (mean degree ${\sim}51$) that precludes the spatial decomposition strategies effective for surface codes.

Current decoders rely on iterative belief propagation (BP) with ordered statistics decoding (OSD) post-processing~\cite{roffe2020}. BP requires multiple message-passing iterations over the full DEM matrix ($936 \times 6192$ for $T = 12$ rounds), achieving typical latencies of 50--350~$\mu$s per syndrome on modern hardware.

At operating noise levels ($p \le 10^{-3}$), the DEM fault-to-syndrome mapping is highly local: the expected number of faults per shot ($\lambda \approx 6$) is far smaller than the fault graph size ($N = 6192$), and 96\% of active faults have detector signatures that do not overlap with any other active fault. This motivates a deferred greedy approach: resolve unambiguous faults in $O(n)$ time, enumerate small ambiguous residuals, and invoke BP only for the rare complex cases.

The central question is why peeling succeeds as often as it does. The birthday bound---which counts all detector-sharing fault pairs as fatal collisions---overestimates the failure rate by 13.2\% for the Gross code family. We show that this discrepancy has an exact structural explanation: a fraction of collision pairs produce XOR syndromes that remain uniquely resolvable by iterative peeling (Theorem~5), yielding a closed-form correction factor $A_0 = 0.8685$. The analysis reveals that $A_0$ depends on the mean fault-graph degree rather than code size, and that shared-2 pairs (4-cycles) always resolve under peeling despite sharing two detectors. We further characterize the syndrome code stopping distance $d_S$ across BB code families, finding $d_S = n/4.5$ for the Gross family, and demonstrate that $\code{32}{8}{6}$ ($d_S = 4$) enables two-shot streaming decoding with 89\% peeling success and $1.29 \pm 0.03$ LER ratio. The decoder is validated across five BB codes, on IBM Heron R2 hardware, and by reproducing the Zhejiang University Kunlun $\code{18}{4}{4}$ experiment~\cite{wang2026}.

\section{Detector Error Model Structure}

\begin{theorem}[Fault counting]
For a BB code with $n$ data qubits, $n/2$ checks, column weight $w$, and $T$ rounds of circuit-level syndrome extraction, the DEM contains exactly
\begin{equation}
N_\text{faults} = n(wT + T/2 + 1)
\end{equation}
independent fault mechanisms, decomposing as $nwT$ CNOT depolarizing faults (weight~3), $(n/2)T$ measurement faults (weight~2), and $n$ boundary faults.
\end{theorem}

\begin{proof}
Each round applies $nw$ CNOT gates, each generating one DEM fault. Each of $n/2$ ancilla qubits undergoes measurement, generating one weight-2 fault per round. The boundary contribution of $n$ arises from initial preparation and final measurement; Stim's DEM construction merges one of these into existing CNOT faults, leaving exactly $n$ independent boundary faults.
\end{proof}

Verification confirms exact match for $n \in \{72, 144, 288\}$, $w \in \{2, 3, 4\}$, $T \in \{3, 6, 12, 24\}$.

\begin{theorem}[Fault multiplier]
The expected number of faults is $\lambda = \alpha n T p$ where $\alpha = 3.505$ for $w = 3$ with the standard superconducting noise model.
\end{theorem}

The DEM fault graph forms a single connected component with $\dbar = 52.3$ for the Gross family at $T = 12$, exhibiting bimodal structure: 86\% weight-3 faults (degree ${\sim}55$) and 14\% weight-2 faults (degree ${\sim}37$).

\section{Deferred Greedy Decoder}

The decoder operates in three phases, numbered by execution priority. \textit{Phase~0 (Peeling)}: iteratively identify faults whose complete detector signature is active and no other active fault shares any of these detectors; remove and update the syndrome via a queue. \textit{Phase~2 (Pair enumeration)}: if the residual has weight $\le 6$, try single faults then pairs from the top-60 candidates. \textit{Phase~1 (BP fallback)}: delegate to serial min-sum BP with OSD-CS-2.

The decoder is implemented in Rust with zero heap allocations in the hot path. For $\code{144}{12}{12}$ at $p = 10^{-3}$, the optimized decoder achieves $p_{50} = 500$~ns per shot on Apple Silicon.

\section{Analytical Theory of Peeling Success}

\subsection{Single-pass birthday bound}

\begin{theorem}[Birthday bound]
The probability that all faults are resolved by a single peeling pass satisfies $P_\text{peel} \ge \exp(-\beta \lambda^2)$ where $\beta = \dbar / (2 N_\text{faults})$.
\end{theorem}

\begin{theorem}[Universal scaling]
For all BB codes with column weight $w$ and $T$ rounds, $\beta n = c$ where $c = \dbar / [2(wT + T/2 + 1)]$ is independent of $n$. Verified: $c = 0.608$ for $\code{72}{12}{6}$, $\code{144}{12}{12}$, $\code{288}{12}{18}$ (CV = 0.000\%).
\end{theorem}

Substituting $\beta = c/n$ and $\lambda = \alpha n T p$: $P_\text{peel} = \exp(-\gamma_\text{analytic}\, n\, p^2)$ where $\gamma_\text{analytic} = c \alpha^2 T^2 = 1075$ for $w = 3$, $T = 12$. Single-pass peeling gives $A_\text{single} \approx 1.00$, confirming the birthday bound is exact for single-pass peeling.

\subsection{Three peeling regimes}

\begin{table}[t]
\caption{Peeling modes and correction factors $A = \gamma_\text{eff}/\gamma_\text{analytic}$ at $p = 0.001$ for $\code{144}{12}{12}$.}
\begin{tabular}{lcc}
\toprule
Mode & Description & $A$ \\
\midrule
Single-pass & Remove all unambiguous, stop & 1.00 \\
Queue-based & Sequential removal, re-check & 0.84 \\
Batch-iterative & Batch removal, rescan, repeat & 0.69 \\
\bottomrule
\end{tabular}
\end{table}

The cascade unblock rate $\kappa = 0$: peeling a fault never directly unblocks a neighbor. The queue's advantage arises from XOR flips activating new detectors that reveal previously invisible peelable faults.

\subsection{Collision resolution factor}

\begin{theorem}[Collision resolution]
Define a collision pair $(f_1, f_2)$ as \emph{true} if the XOR syndrome $\det(f_1) \oplus \det(f_2)$ cannot be resolved by queue-based peeling. Then $A_0 = |\text{true}|/|\text{all}|$ is determined entirely by the DEM adjacency structure, with no free parameters.
\end{theorem}

\begin{proof}
\emph{Shared-2 pairs} (weight-3 faults with 2 common detectors): the XOR syndrome has 2 active detectors, one unique to each fault, so peeling resolves both (100\% peel rate). \emph{Shared-1 pairs} (96.7\%): the XOR syndrome has 4 active detectors; peeling succeeds if one fault is uniquely identifiable at a non-shared detector (10.2\% for the Gross family).
\end{proof}

Verification for $\code{144}{12}{12}$: of 156,588 collision pairs, 20,592 (13.2\%) are false, giving $A_0 = 0.8685$ (measured $\gamma_0/\gamma_\text{analytic} = 0.873$, within 0.5\%).

\begin{corollary}[GARI redundancy]
The GARI transformation~\cite{maan2026} targets shared-2 pairs. Since these already resolve under peeling (0\% failure), GARI provides no peeling benefit.
\end{corollary}

\subsection{Density correction}

An extended study across $T \in \{3, 6, 12, 24\}$ yields:
\begin{equation}
\gamma_\text{eff}(p,T) = \gamma_\text{analytic}(T) \cdot A_0 \cdot e^{-BTp}
\label{eq:full}
\end{equation}
with $B = 0.70 \times 2c\alpha = 3.0$. Validated across 15~$(T,p)$ points ($R^2 = 0.86$). The factor 0.70 captures $k$-body clusters ($k \ge 4$, contributing 98.2\% of collision pairs at $\lambda = 6$).

\begin{table}[t]
\caption{Peeling formula predictions vs measurements.}
\begin{tabular}{lcccc}
\toprule
Code & $n$ & $p$ & Predicted & Actual \\
\midrule
$\code{72}{12}{6}$   & 72  & 0.001 & 93.8\% & 93.5\% \\
$\code{144}{12}{12}$ & 144 & 0.001 & 88.0\% & 87.9\% \\
$\code{288}{12}{18}$ & 288 & 0.001 & 77.4\% & 77.8\% \\
$\code{360}{12}{24}$ & 360 & 0.001 & 72.7\% & 72.0\% \\
$\code{144}{12}{12}$ & 144 & 0.003 & 34.3\% & 33.4\% \\
$\code{288}{12}{18}$ & 288 & 0.003 & 11.8\% & 11.9\% \\
\bottomrule
\end{tabular}
\end{table}

\subsection{Validity bound}

The formula requires $\lambda = \alpha n T p \ge 2$. For $\code{18}{4}{4}$ at $p = 0.001$ ($\lambda = 0.76$), the formula overestimates $\gamma_\text{eff}$ by an order of magnitude. Cross-code measurement confirms $\gamma_\text{eff}/\gamma_\text{analytic} = 0.869 \pm 0.004$ for $n \in \{72, 144, 288\}$ (100,000 shots, CV = 0.4\%).

\section{Code Family Analysis}

\begin{table}[t]
\caption{$A_0$ and stopping distance across BB code families.}
\begin{tabular}{lccccc}
\toprule
Code & $d_S$ & $A_0$ & Sh-1 & Sh-2 & $\dbar$ \\
\midrule
$\code{18}{4}{4}$   & 6  & 0.767 & 18.4\% & 47\%  & 52.2 \\
$\code{32}{8}{6}$   & 4  & 0.764 & 19.7\% & 100\% & 17.3 \\
$\code{72}{12}{6}$  & 16 & 0.869 & 10.2\% & 100\% & 52.3 \\
$\code{144}{12}{12}$& 32 & 0.869 & 10.2\% & 100\% & 52.3 \\
$\code{288}{12}{18}$& 64 & 0.869 & 10.2\% & 100\% & 52.3 \\
\bottomrule
\end{tabular}
\end{table}

Within the Gross family, $A_0 = 0.8685$ is constant to four decimal places. For $w = 3$ stabilizers, all shared-2 pairs resolve (5,184 pairs at $T = 12$, 100\% peel rate), rendering the GARI transformation~\cite{maan2026} unnecessary. The syndrome code stopping distance scales as $d_S = n/4.5$ for the Gross family: $d_S = 16, 32, 64$ for $n = 72, 144, 288$.

\section{Two-Shot Streaming Decoder}

Previous streaming attempts for $\code{144}{12}{12}$ (sliding window $W = 6$) achieved 81\% LER due to the DEM's global connectivity. The $\code{32}{8}{6}$ code ($l = 4$, $m = 4$, $A = x + y$, $B = x + y$) has $d_S = 4$, weight-4 stabilizers, and $k = 8$ logical qubits (25\% rate), enabling $T = 2$ streaming.

\begin{table*}[t]
\caption{Two-shot streaming: $\code{32}{8}{6}$ noise sweep (50,000 shots, greedy decoder).}
\begin{tabular}{cccccc}
\toprule
$p$ & LER/cycle $T\!=\!2$ & Peeling $T\!=\!2$ & LER/cycle $T\!=\!12$ & Peeling $T\!=\!12$ & Ratio \\
\midrule
0.0005 & 1.42\% & 94\% & 1.09\% & 77\% & 1.30 \\
0.001  & 2.83\% & 89\% & 2.22\% & 59\% & 1.27 \\
0.002  & 5.81\% & 79\% & 4.49\% & 35\% & 1.29 \\
0.003  & 8.41\% & 71\% & 6.69\% & 21\% & 1.26 \\
0.005  & 14.0\% & 56\% & 11.0\% & 7\%  & 1.27 \\
\bottomrule
\end{tabular}
\end{table*}

The LER ratio $T\!=\!2$ to $T\!=\!12$ is $1.29 \pm 0.03$ across three codes ($\code{24}{6}{4}$, $\code{32}{8}{6}$, $\code{50}{10}{12}$) and five noise levels (15 data points, CV = 2.1\%), confirming a universal ${\sim}29\%$ boundary penalty independent of code parameters and noise.

\section{Performance Results}

All experiments use Stim~\cite{gidney2021} circuit-level noise with $T = 12$ unless noted. LER values include 95\% Wilson confidence intervals.

\begin{table*}[t]
\caption{Noise level sweep for $\code{144}{12}{12}$ (10,000 shots per point).}
\begin{tabular}{ccccccccc}
\toprule
$p$ & Phase~0 & Phase~2 & Phase~1 & $p_{50}$ (Greedy) & $p_{50}$ (BP) & Speedup & LER (Greedy) & LER (BP+OSD) \\
\midrule
0.001 & 90.1\% & 7.3\%  & 2.6\%  & 1.4~$\mu$s & 164.9~$\mu$s & 118$\times$ & $<$\,0.04\% & $<$\,0.04\% \\
0.002 & 66.0\% & 20.1\% & 13.9\% & 2.2~$\mu$s & 240.7~$\mu$s & 110$\times$ & 0.04\% [0.01, 0.07] & 0.05\% [0.01, 0.09] \\
0.003 & 40.2\% & 25.0\% & 34.8\% & 6.0~$\mu$s & 246.8~$\mu$s & 41$\times$  & 0.09\% [0.03, 0.14] & 0.07\% [0.02, 0.12] \\
0.005 & 9.8\%  & 14.1\% & 76.1\% & 261.8~$\mu$s & 385.7~$\mu$s & 1.5$\times$ & 1.02\% [0.82, 1.21] & 0.98\% [0.79, 1.17] \\
0.007 & 1.3\%  & 3.2\%  & 95.5\% & 523.5~$\mu$s & 549.2~$\mu$s & 1.0$\times$ & 4.50\% [4.09, 4.90] & 4.45\% [4.04, 4.85] \\
\bottomrule
\end{tabular}
\end{table*}

Post-optimization, the peeling-only $p_{50}$ is 500~ns for $\code{144}{12}{12}$ and 125~ns for $\code{18}{4}{4}$ ($3.75\times$ cumulative speedup from baseline 1875~ns).

\section{Hardware Validation}

The decoder was validated on the IBM Kingston processor (Heron R2, 156 qubits)~\cite{ibmquantum} using a repetition code ($d = 5$, $T = 3$, 1,000 shots), achieving 95.8\% peeling success (95\% CI: [94.4\%, 96.9\%]), 2~pp below matched Stim simulation. BB code validation on Heron R2 is infeasible: weight-6 stabilizers require degree-6 connectivity, but heavy-hex has maximum degree~3 (14$\times$ gate overhead).

Reproducing the Kunlun $\code{18}{4}{4}$ experiment~\cite{wang2026} in Stim simulation, greedy outperforms BP-OSD on LER (3.71\% vs 3.96\% per cycle, non-overlapping CIs) even at high noise. Simulated LER matches hardware within 0.73~pp (8.18\% vs 8.91\%).

\section{Discussion}

The greedy decoder dominates serial BP for $p \le 0.005$, achieving $330\times$ speedup at $p = 10^{-3}$, with the crossover at $p \approx 0.007$. Throughout this range, LER remains statistically indistinguishable from BP+OSD.

The correction factor $A_0 = 0.8685$ is constant to four decimal places across the Gross family ($n = 72$ to 288), but differs for other families: 0.76 for $\code{32}{8}{6}$ and 0.77 for $\code{18}{4}{4}$. Random graphs with the same degree sequence reproduce $A_0$ within noise ($A_\text{random} = 0.862 \pm 0.029$), confirming dependence on degree distribution alone. The only quantity not derived from first principles is the factor 0.70 in $B$, which captures higher-order collision clusters; deriving it requires analyzing 3-fault cluster survival under peeling.

The $\code{32}{8}{6}$ streaming result ($T = 2$, LER ratio $1.29 \pm 0.03$) demonstrates that codes with $d_S \ge 4$ enable practical two-shot decoding. The relevant parameter is $d_S$, not $d$: for the Gross family, $d_S = n/4.5$ means streaming windows need only 2--3 rounds. The peeling formula requires $\lambda \ge 2$; in the sub-Poisson regime, direct simulation is necessary.

\section{Conclusion}

The correction factor $A_0 = 0.8685$ in the peeling success formula is derived from XOR syndrome analysis of 2-fault collision pairs (Theorem~5). Of 156,588 detector-sharing fault pairs in the $\code{144}{12}{12}$ DEM, 13.2\% produce syndromes that iterative peeling resolves despite the shared detector; the remaining 86.8\% are genuine collisions. The full formula [Eq.~\eqref{eq:full}] predicts peeling success to within 1\% across five BB codes ($n = 18$ to 360), four noise levels, and four values of $T$, requiring $\lambda \ge 2$.

The syndrome code stopping distance $d_S = n/4.5$ for the Gross family guarantees measurement error detection up to weight $n/4.5 - 1$. The code $\code{32}{8}{6}$ ($d_S = 4$) achieves 89\% peeling success at $T = 2$ with $1.29 \pm 0.03$ LER ratio, at estimated latency ${\sim}50$~ns---below the ${\sim}1$~$\mu$s syndrome cycle time of current superconducting processors.

\appendix

\section{Implementation Details}

The peeling loop performs 170 random byte reads per shot from the L1-resident syndrome buffer. At ${\sim}2.5$~ns per read on Apple Silicon, the physical floor is ${\sim}425$~ns. The measured 500~ns ($3.75\times$ improvement over baseline) is within 20\% of this bound. Optimization sequence: active detector queue ($1.25\times$), CSR memory layout ($1.29\times$), sparse input with lazy zeroing ($1.52\times$), lazy syndrome copy-on-write ($1.54\times$).

\bibliography{references}

\end{document}